\documentclass{amsart}
\usepackage[mathscr]{eucal}
\usepackage{amsmath,amssymb,hyperref}
\usepackage{amstext}
\usepackage{graphics}
\usepackage{xcolor}
\usepackage{mathrsfs}
\usepackage{epstopdf}

\usepackage{tikz}
\usetikzlibrary{matrix}

\usepackage[all]{xy}

\usepackage[T1]{fontenc}
\usepackage[utf8]{inputenc}

\newtheorem{thm}{Theorem}[section]
\newtheorem{THM}{Theorem}

\newtheorem{cor}[thm]{Corollary}
\newtheorem{prop}[thm]{Proposition}

\newtheorem{lemma}[thm]{Lemma}

\theoremstyle{definition}

\newtheorem{definition}[thm]{Definition}
\newtheorem{remark}[thm]{Remark}
\newtheorem{example}[thm]{Example}

\def\X0{X^{\circ}}

\def\Y0{Y^{\circ}}

\input xy
\xyoption{all}

\addtocounter{section}{0}             
\numberwithin{equation}{section}       

\begin{document}

\title[Positive Characteristic Darboux-Jouanolou Integrability of Differential Forms]
{Positive Characteristic Darboux-Jouanolou Integrability of Differential Forms}


\author[E. A. Santos]{Edileno de Almeida SANTOS}
\address{Instituto de Ciência, Engenharia e Tecnologia (ICET) - Universidade Federal dos Vales do Jequitinhonha e Mucuri (UFVJM), Campus do Mucuri, Teófilo Otoni - MG,
Brazil}
\email{edileno.santos@ufvjm.edu.br}

\author[S. Rodrigues]{Sergio RODRIGUES}
\address{Faculdade de Ciências Exatas e Tecnologia (FACET) - Universidade Federal da Grande Dourados (UFGD), Rodovia Dourados - Itahum, Km 12 - Cidade Universitátia, Dourados - MS,
Brazil}
\email{sergiorodrigues@ufgd.edu.br}

\subjclass{37F75} \keywords{Positive Characteristic, Vector Fields, Codes}


\begin{abstract}

We prove a Darboux-Jouanolou type theorem on the algebraic integrability of polynomial differential $r$-forms over arbitrary fields ($r\geq 1$). We also investigate the Darboux's method for producing integrating factors.

\end{abstract}

\maketitle

\setcounter{tocdepth}{1}
\sloppy



\section{Introduction}
In his seminar memoir \cite{Darboux}, G. Darboux (1878) showed the fascinating relationships between integrability and the existence of algebraic solutions for a planar polynomial differential system. The classical approach of Darboux shows that, for a planar polynomial vector field of degree $d$, on $\mathbb R^2$ or $\mathbb C^2$, from $\binom{d+1}{2}+1$ invariant algebraic curves we can deduce an analytic (possibly multivaluated) first integral. This method was studied by H. Poincaré (1891) in \cite{Poincare}, where he observes the difficulty to obtain algorithmically invariant curves.

A general result of J.-P. Jouanolou (1979)in \cite{Jouanolou} shows that if $K$ is a field of characteristic $0$ and $\omega$ is a polynomial $1$-form of degree $d$ on $K^n$ admitting at least $\binom{d-1+n}{n}\cdot \binom{n}{2}+2$ irreducible invariant  algebraic hypersurfaces, then $\omega$ has a rational first integral, computed in terms of the invariant hypersurfaces.

The {\it Darboux Integration Method} have been successfully used in Physics (see \cite{Hewitt}, \cite{Llibre1}, \cite{Llibre2}, \cite{Valls}, \cite{Zhang}). For example, C. G. Hewitt (1991) in \cite{Hewitt} study some new solutions to the Einstein field equations.

For holomorphic foliations on compact complex manifolds, E. Ghys (2000) in \cite{Ghys} gives an extended version of Jouanolou's Theorem. Over fields of characteristic zero, a Darbou-Jouanolou type theorem is proved in \cite{Correa} for polynomial differential $r$-forms, $r\geq 1$. 

In positive characteristic, M. Brunella and M. Nicollau (1999) proved in \cite{B} that if $\omega$ is a rational $1$-form on a smooth projective variety over an algebraically closed field $K$ of positive characteristic $p>0$ with infinitely many invariant hypersurfaces, then $\omega$ admits a rational first integral.

In sharp contrast with the characteristic $0$ case, where a theorem of Jouanolou says that a generic vector field on the complex plane does not admit any invariant algebraic curve, J. V. Pereira (2001) shows in \cite{JVP} that a generic vector field on an affine space of positive characteristic admits an invariant algebraic hypersurface (the generic condition is that the divergent of the vector field must be zero).

Our goal in this paper is to extend results of \cite{Correa} ($r$-forms over characteristic $0$, $r>0$) and \cite{Santos} ($1$-forms over arbitrary fields) giving a general account of Darboux-Jouanolou integrability for polynomial $r$-forms on $K^n$ ($r\geq 1$), where $K$ is an arbitrary field. We obtain a general version of {\it Darboux-Jouanolou Criterion} (\cite{Jouanolou},  Théorème 3.3, p. 102): for each polynomial $r$-form $\omega$ of degree $d$ on $K^n$, with $K$ of characteristic $p\geq 0$, we define a natural number $N_K(n,d,r)$ (see Definition \ref{D:NK}) that depends only on $n$, $r$, $d$ and $p$, and we prove the following theorem.

\begin{THM}[Theorem \ref{T:Main}]
Let $\omega\in \Omega^r_d(K^n)$ be a polynomial $r$-form of degree $d$ over an arbitrary field $K$ and $N=N_K(n,d-1,r+1)$. If $\omega$ has $N+r+1$ invariant hypersurfaces, then $\omega$ admits a rational first integral.
\end{THM}

The number $N_K(n,d,r)$  has the property that in characteristic $0$ holds $N_K(n,d,r)= \binom{d+n}{n}\cdot \binom{n}{r}$ but in characteristic $p>0$ we have only $N_K(n,d,r) \leq \binom{d+n}{n}\cdot \binom{n}{r}$.

Another approach is to look at {\it integrating factors}, that is, a function $f$ for a given  $r$-form $\omega$ such that $d(f\omega)=0$. In this direction we obtain the result below.

\begin{THM}[Theorem \ref{T:Main2}]
Let $\omega\in \Omega^r_{K[z]/K}$ be a polynomial $r$-form of degree $d$ and $N=N_K(n,d-1,r+1)$. If $\omega$ admits $N+r$ invariant hypersurfaces, then $\omega$ has a rational first integral or a rational integration factor.
\end{THM}

Our exposition was heavily influenced by the excellent monograph \cite{JVP3} of J. V. Pereira, where the Darboux-Jouanolou theory is appreciated on $\mathbb C^2$.




\section{Differential Forms}
We remember first some definitions and results from reference \cite{Santos}, where the reader can obtain more information about vector fields and differential forms over arbitrary fields.

\begin{definition}
Let $K$ be a field of arbitrary characteristic $p>0$, where $p$ is a prime integer. The {\it ring of differential constants} is the sub-ring of $K[z]:=K[z_1,..., z_n]$ given by
$$
K[z^p]:=K[z_1^p,...,z_n^p]
$$
and the {\it field of differential constants} is the sub-field of $K(z):=K(z_1,..., z_n)$ given by
$$
K(z^p):=K(z_1^p,...,z_n^p)
$$
We have natural inclusions $K[z]\subset K(z)$ and $K[z^p]\subset K(z^p)$. In fact, we have more: $K[z^p]$ is a  $K$-sub-module of $K[z]$ generated by $\{G^p: G\in K[z] \}$ and $K(z^p)$ is a $K$-vector subspace of $K(z)$ generated by $\{g^p: g\in K(z) \}$. The elements of $K(z^p)$ are called {\it $\partial$-constants}.

\end{definition}

The ring of polynomials $K[z]$ is a $K[z^p]$-algebra. If $n\geq 0$ then, for each $ j=1$,..., $n$, $ z_j^n= z_j^{p q_j+ i_j}=(z_j^{q_j})^pz_j^{i_j}$ with $ 0 \leq i_j < p$, hence $ K[z]$  is finitely generated as a $K[z^p]$-algebra. In the next proposition we will prove that $ K(z)$ is also  finitely generated as a vector space over $K(z^p)$.

Note that $K[z^p]-K^*$ is a $K$-sub-module of the module $K[z]$, and we have the following isomorphisms
$$
\mathcal M := K[z]\otimes_{K[z^p]} K[z^p]\cong \frac{K[z]}{(K[z^p]-K^*)}\otimes_K K[z^p] \cong \bigoplus_{0\leq i_1,...,i_n\leq p-1} K[z^p]\cdot z_1^{i_1}...z_n^{i_n}
$$
then $K[z]\simeq \mathcal M$. To make a $K(z^p)$-algebra from $K[z]$ we define
$$
\mathcal L := K[z]\otimes_{K[z^p]} K(z^p)\cong \frac{K[z]}{(K[z^p]-K^*)}\otimes_K K(z^p) \cong \bigoplus_{0\leq i_1,...,i_n\leq p-1} K(z^p)\cdot z_1^{i_1}...z_n^{i_n}
$$



\begin{prop}[see also \cite{Brion}, 1.1.1 Lemma, p. 3]
Let $K$ be a field of characteristic $ p>0$. Then, as $K(z^p)$-vector spaces, $K(z)$ is isomorphic to $\mathcal L$, and hence $ K(z)$ is a finite dimensional vector space over $K(z^p)$, and $\dim_{K(z^p)}(K(z))=p^n$.
\end{prop}
\begin{proof}
A base for $K(z)$ over $K(z^p)$ is given by the monomials $z_1^{i_1}...z_n^{i_n}$, where $0\leq i_1,...,i_n\leq p-1$. In fact, if $f=\frac{P}{Q}\in K(z)$, we can write
$$
f=\frac{P}{Q}=\left(\frac{1}{Q}\right)^p\cdot Q^{p-1}P
$$
and the polynomial $Q^{p-1}P$ can be expressed as a $K(z^p)$-linear combination of the monomials $z_1^{i_1}...z_n^{i_n}$.

\end{proof}


In what follows we consider multi-indexes $I=(i_1,...,i_n)$. If $J=(j_1,...,j_n)$ is another multi-index of {\it length} $n$, we write $J<I$ (resp. $J\leq I$) if $j_1<i_1$,..., $j_n<i_n$ (resp. $j_1\leq i_1$,..., $j_n\leq i_n$). There are two special multi-indexes: $\overrightarrow{0}=(0,...,0)$ and $\overrightarrow{p}=(p,...,p)$.

\begin{prop}\label{P:Base}
Let $ K$ be a field of characteristic $p>0$. Then
$$
K(z^p)=\{f\in K(z): df=0 \}
$$
and
$$
K[z^p]=\{F\in K[z]: dF=0 \}
$$
\end{prop}

\begin{proof}
A base for $ K[z]$ over $K[z^p]$ is given by the monomials $z_1^{i_1}...z_n^{i_n}$, where $0\leq i_1,...,i_n\leq p-1$, and this is also a base  for $ K(z)$ over $K(z^p)$, by the above proposition.
It is clear that $K(z^p)\subset \{f\in K(z): df=0 \}$. Let $f\in K(z)$ be a rational function then, by the above proposition, 
$$
f(z)=\sum_{\overrightarrow{0}\leq I <\overrightarrow{p}} g_I(z^p)z^I
$$
Hence, writing $I=(I(1),...,I(n))$,
$$
\frac{\partial f}{\partial z_j}=\sum_{\begin{array}{c}
\overrightarrow{0} \leq I \leq \overrightarrow{p} \\
I(j)>0
\end{array}} g_I(z^p)I(j)z^Iz_j^{-1}
$$
and, fixing $j$, the monomials $z^Iz_j^{-1}$ are $K(z^p)$-linearly independent. Therefore, if $df=0$, then all $g_I$ with $\overrightarrow{0}< I$ must be $0$, so we must have $f=g_{\overrightarrow{0}}(z^p)\in K(z^p)$. In particular, $ K[z^p]=\{F\in K(z): dF=0 \}$. 

\end{proof}


Define the $K[z^p]$-sub-module of $K[z^p]$
$$
\mathcal M_d:= \bigoplus_{\begin{array}{c}
0 \leq i_1,...,i_n \leq p-1 \\
0\leq i_1+...+i_n \leq d
\end{array}} K[z^p]\cdot z_1^{i_1}...z_n^{i_n}
$$
and the $K(z^p)$-vector subspace of $K(z)$
$$
\mathcal L_d:= \bigoplus_{\begin{array}{c}
0 \leq i_1,...,i_n \leq p-1 \\
0\leq i_1+...+i_n \leq d
\end{array}} K(z^p)\cdot z_1^{i_1}...z_n^{i_n}
$$
where the number $d$ is called the {\it $p$-degree} of the module (or vector space) if $d\leq n(p-1)$.

We can write a decomposition by $p$-degree
$$
K[z^p]=\mathcal M_0 \subset \mathcal M_1\subset ... \subset \mathcal M_{n(p-1)}=K[z]
$$
and
$$
K(z^p)=\mathcal L_0 \subset \mathcal L_1\subset ... \subset \mathcal L_{n(p-1)}=K(z)
$$

Let $A$ be a ring and let $B$ be an $A$-algebra. We denote by $\Omega^1_{B/A}$ the {\it module of relative differential $1$-forms of $B$ over $A$}, that is, if $E$ is the free $A$-module generated by $ db $ with $b\in B$ and $ F$ is the submodule generated by $da$ with $ a \in A$, then $\Omega^1_{B/A}:= E/F$ with an $A$-derivation $d: B \rightarrow \Omega^1_{B/A}$ (see \cite{Liu}, page 210 ). By using the exterior product, we obtain $\Omega^r_{B/A}:=\Omega^1_{B/A}\wedge...\wedge \Omega^1_{B/A}$ ($r$ times). If $M\subset B$ is an $A$-sub-module of $B$, then we define
$$
\Omega^1_{M/A}:=d(M)\subset \Omega^1_{B/A}
$$

We can easily deduce the isomorphisms  $\Omega^1_{K[z]/K}\simeq \Omega^1_{K[z]/K[z^p]}$ as $K[z^p]$-modules and $\Omega^1_{K(z)/K}\simeq \Omega^1_{K(z)/K(z^p)}$ as $K(z^p)$-vector spaces.

\begin{definition}\label{D:NK}
If $ d \geq 0$ is an integer, then we denote by $\Omega^r_d(K^n)\subset \Omega^r_{K[z]/K}$ the $K$-vector subspace of polynomial differential $r$-forms of degree $\leq d$. Then
$$
\Omega_{\mathcal M_d/K[z^p]}^r:=\Omega^1_{\mathcal M_d/K[z^p]}\wedge
...\wedge\Omega^1_{\mathcal M_d/K[z^p]}\simeq \Omega_0^r(K^n)\otimes_K \mathcal M_d= \Omega_0^r(K^n)\otimes_K \mathcal M_{s(d)}
$$
and
$$
\Omega_{\mathcal L_d/K(z^p)}^r:=\Omega_{\mathcal L_d/K(z^p)}^1\wedge
...\wedge \Omega_{\mathcal L_d/K(z^p)}^1\simeq \Omega_0^r(K^n)\otimes_K \mathcal L_d= \Omega_0^r(K^n)\otimes_K \mathcal L_{s(d)}
$$
where $s(d)=\min\{d,n(p-1)\}$, and we define
$$
N_K(n, d, r)=N(p,n,d,r):=\dim_{K(z^p)}(\Omega_{\mathcal L_d/K(z^p)}^r)  \leq \min\{\binom{n}{r}\cdot p^n, \dim_K(\Omega_d^r(K^n))\}
$$

\end{definition}

These two modules of $r$-forms have natural decomposition by $p$-degree:
$$
\Omega^r_{\mathcal M_0/K[z^p]}\subset \Omega^r_{\mathcal M_1/K[z^p]}\subset ... \subset \Omega^r_{\mathcal M_{n(p-1)}/K[z^p]}=\Omega^r_{K[z]/K}
$$
and
$$
\Omega^r_{\mathcal L_0/K(z^p)}\subset \Omega^r_{\mathcal L_1/K(z^p)}\subset ... \subset \Omega^r_{\mathcal L_{n(p-1)}/K(z^p)}=\Omega^r_{K(z)/K}
$$


For every $d\geq 0$ we have a natural inclusion ${\Omega}_d^r(K^n)\rightarrow\Omega^r_{\mathcal L_{s(d)}/K(z^p)}$.



\section{Logarithmic 1-forms and Residues}

\begin{definition}
Let $F_1$,..., $F_m$ be a collection of irreducible polynomials and $\lambda_1$,..., $\lambda_m\in K(z^p)$ be $\partial$-constants. The linear combination
$$
\eta=\lambda_1\frac{dF_1}{F_1}+...+\lambda_m\frac{dF_m}{F_m}
$$
will be called a {\it $\partial$-logarithmic} $1$-form (or simply a {\it logarithmic} $1$-form). Note that every $\partial$-logarithmic $1$-form $\eta$ is {\it closed}, that is, $d\eta=0$.
\end{definition}

Remember also the {\it Hilbert's Nullstellensatz}:

\begin{thm}[\cite{Lang}, Theorem 1.5, p. 380]
Let $I$ be an ideal of $K[z]$ and let $V(I)=\{c\in (\overline{K})^n: f(c)=0, \forall f\in I\}$ be the algebraic variety associated to $I$, where $\overline{K}$ is the algebraic closure of $K$. Let $P$ be a polynomial in $K[z]$ such that $P(c)=0$ for every zero $(c)=(c_1,...,c_n)\in V(I)$. Then there is an integer $m>0$ such that $P^m\in I$.
\end{thm}

Using the above theorem, Jouanolou proved the following important lemma (\cite{Jouanolou}, Lemme 3.3.2, p. 102) in the case that the characteristic of the field $ K$ is  0.

\begin{lemma}\label{L:Jouanolou}
If $\mathcal S$ is a finite representative system of primes in $K[z]=K[z_1,...,z_n]$ (that is, $\mathcal S$ is a finite collection of distinct irreducible polynomials), then the $K$-linear map
$$
K^{(\mathcal S)}\longrightarrow \Omega^1_{K(z)/K}
$$
$$
(\lambda_j)_{F_j\in \mathcal S}\longmapsto \sum_{F_j\in \mathcal S} \lambda_j \frac{dF_j}{F_j}
$$
is injective.
\end{lemma}

\begin{remark}
Two irreducible polynomials $P$, $Q\in K[z]$ are {\it distinct} if there is no constant $\lambda \in K^*$ such that $P=\lambda Q$, that is, the rational function $f=\frac{P}{Q}\in K(z)$ is not a constant ($f\notin K$).
\end{remark}

The following lemma is a residue formula over positive characteristic $p>0$.

\begin{lemma}[\cite{Santos}, Lemma 5.2]\label{L:Residuo}
Let $G\in K[x]$, $G(0)\neq 0$, be a polynomial function in one variable, where $K$ is a field of positive characteristic $p>0$. If $\alpha \in K(x^p)$, then
$$
Res(\alpha\cdot \frac{dG}{G},0)=0
$$
\end{lemma}

The following lemma is the {\it Positive Characteristic Jouanolou's Lemma}.

\begin{lemma}[\cite{Santos}, Lemma 5.3]\label{L:Logaritmicas}
If $\mathcal S$ is a finite representative system of non-$\partial$-constant primes in $K[z]-K[z^p]$ (that is, $\mathcal S$ is a finite collection of distinct non-$\partial$-constant irreducible polynomials), then the $K(z^p)$-linear map
$$
K(z^p)^{(\mathcal S)}\longrightarrow \Omega^1_{K(z)/K}
$$
$$
(\lambda_j)_{F_j\in \mathcal S}\longmapsto \sum_{F_j\in \mathcal S} \lambda_j \frac{dF_j}{F_j}
$$
is injective.
\end{lemma}

\begin{remark}
If $K$ is a perfect field, then every irreducible polynomial in $K[z]$ is non-$\partial$-constant. Now let $K$ be a non perfect field, that is, there is $\lambda \in K$ such that $\lambda$ is not a $p$-power, where $p>0$ is the characteristic of $K$. Then $P(x):=x^p-\lambda$ is an irreducible polynomial in $K[x]$, but $dP=0$.
\end{remark}





\begin{definition}[\cite{Godbillon}, 2.4 Définition, p. 116]
Let $\omega\in \Omega^r_{K(z)/K}$ be a rational $r$-form. We define the space of {\it tangent $1$-forms} of $\omega$ by
$$
\mathcal E^*(\omega)=\{\eta \in \Omega^1_{K(z)/K}: \omega\wedge\eta=0 \}
$$
\end{definition}

\begin{lemma}[\cite{Correa}, Lemma 2.1] \label{L:Base}
Let $\omega\in \Omega^r_{K(z)/K}$. If there are $r$ elements $\eta_1$,..., $\eta_r\in \mathcal E^*(\omega)$ linearly independent over $K(z)$, then there is $h\in K(z)$ such that
$$
\omega=h\cdot\eta_1\wedge... \wedge \eta_r
$$
\end{lemma}
\begin{proof}
We can complete $\{\eta_1,..., \eta_r\}$ to make a $K(z)$-basis $\{\eta_1,..., \eta_r,\eta_{r+1}...,\eta_{M}\}$ for $\Omega^r_{K(z)/K}$, where $M=\dim_{K(z)} \Omega^r_{K(z)/K}$. So we can write
$$
\omega=\sum_{1\leq i_1<...<i_r\leq M} h_{i_1...i_r}\cdot \eta_{i_1}\wedge... \wedge \eta_{i_r}
$$

Since $\omega\wedge \eta_i=0$ for $i=1$,..., $r$, we have $\omega=h_{1...r}\cdot \eta_{1}\wedge... \wedge \eta_{r}$.

\end{proof}


\section{Darboux Method}

\begin{definition}
Let $F\in K[z]$ and $\omega\in \Omega_d^r(K^n)$ be a polynomial $r$-form of degree $ d$. We say that $F$ is {\it invariant} by $\omega$ if $dF\neq 0$ and $F$ divides $\omega\wedge dF$, that is,
$$
\omega\wedge dF=F\cdot \Theta_F
$$
where $\Theta_F\in \Omega_{d-1}^{r+1}(K^n)$ is called the {\it associated cofactor} of $F$. If $\overline{K}$ is the algebraic closure of $K$, by Hilbert's Nullstellensatz this is equivalent to saying that
$$
(\omega\wedge dF) \mid_V=0
$$
where $V=\{F=0\}$ is the associated algebraic hypersurface in $(\overline{K})^n$. 

A rational function $f$ is said to be a  {\it rational first integral} for $\omega$ if $df\neq 0$ and $\omega\wedge df=0$, that is, $ df$  is tangent do $ \omega$.

\end{definition}

The following proposition is similar to Proposition 6.1 of  \cite{Santos}.

\begin{prop}\label{P:Log}
Let $\omega\in \Omega^r_{K[z]/K}$ be an $r$-form, where $K$ is an arbitrary field. If there are invariant irreducible polynomials $F_1$,..., $F_m$ in $K[z]$ and $\partial$-constants $\lambda_1$,..., $\lambda_m$ in $K(z^p)^*$ such that
$$
\lambda_1\Theta_{F_1}+...+\lambda_m\Theta_{F_m}=0
$$
then there is a logarithmic $1$-form $\eta\neq 0$ such that $\omega\wedge\eta=0$. 
\end{prop}
\begin{proof}
We can consider the logarithmic $1$-form
$$
\eta=\sum_{i=1}^{m} \lambda_i\cdot \frac{dF_i}{F_i}
$$
Therefore
$$
\omega\wedge\eta=\sum_{i=1}^{m} \lambda_i\omega\wedge \frac{dF_i}{F_i}=\sum_{i=1}^{m} \lambda_i\Theta_{F_i}=0
$$

\end{proof}

\begin{cor}\label{Corol P:log}
Let $\omega\in \Omega_d^r(K^n)$ be a polynomial $r$-form of degree $d$ and $N=N_K(n,d-1,r+1)=\dim_{K(z^p)}(\Omega^{r+1}_{\mathcal L_{d-1}/K(z^p)})$. If $\omega$ has $N+1$ invariant irreducible polynomials, then there is a logarithmic $1$-form $\eta\neq 0$ such that $\omega\wedge\eta=0$.
\end{cor}
\begin{proof}
By hypothesis, $\omega$ has $F_1$,..., $F_m$ as invariant irreducible polynomials, where $m=N+1$. The $m$ associated cofactors $\Theta_{F_1}$,..., $\Theta_{F_m}$ are elements of $\Omega^{r+1}_{\mathcal L_{d-1}/K(z^p)}= \Omega_0^{r+1}(K^n)\otimes_K \mathcal M_{s(d)}$ and since $m=\dim_{K(z^p)}(\Omega^{r+1}_{\mathcal L_{d-1}/K(z^p)})+1$ then there exist $\lambda_1$,..., $\lambda_m \in K(z^p)$, not all zero, such that
$$
\lambda_1 \Theta_{F_1}+...+\lambda_m \Theta_{F_m}=0
$$
By Proposition \ref{P:Log},  the logarithmic $1$-form $\eta:=\lambda_1 \frac{dF_1}{F_1}+...+\lambda_m \frac{dF_m}{F_m}$ satisfies
$$
\omega\wedge \eta=\lambda_1 \Theta_{F_1}+...+\lambda_m \Theta_{F_m}=0
$$

\end{proof}

\begin{example}
Consider the $2$-form $\omega=\alpha x dy\wedge dz+\beta y dz\wedge dx+\gamma z dx\wedge dy$, where $\alpha$, $\beta$, $\gamma$ are $\partial$-constants in $K(z^p)^*$. We can easily see that $F=x$, $G=y$ and $H=z$ are invariant. The associated cofactors are $\Theta_x=\omega\wedge \frac{dx}{x}=\alpha dx\wedge dy \wedge dz$, $\Theta_y=\omega\wedge \frac{dy}{y}=\beta dx\wedge dy \wedge dz$ and $\Theta_z=\omega\wedge \frac{dz}{z}=\gamma dx\wedge dy \wedge dz$. If $\lambda_1$, $\lambda_2$, $\lambda_3$ are $\partial$-constants not all zero satisfying $\alpha\cdot \lambda_1+\beta\cdot \lambda_2+\gamma\cdot \lambda_3=0$ then
$$
\lambda_1\Theta_x+\lambda_2\Theta_y+\lambda_3\Theta_z=0
$$
and the  logarithmic $1$-form $\eta=\lambda_1 \frac{dx}{x}+\lambda_2 \frac{dy}{y}+\lambda_3 \frac{dz}{z}$ is tangent to $\omega$.
\end{example}

\begin{definition}
Let $\omega\in \Omega^r_{K(z)/K}$ be a rational $r$-form. The {\it $p$-degree} of $\omega$ is the smaller $\epsilon$ such that $\omega\in \Omega_{\mathcal L_{\epsilon}/K(z^p)}^{r}$.
\end{definition}

\begin{remark}
In Corollary  \ref{Corol P:log} it is not possible to exchange  
$\Omega_{\mathcal L_{d}/K(z^p)}^{r}=\Omega_{\mathcal L_{s(d)}/K(z^p)}^{r}$ by $\Omega_{\mathcal L_{\epsilon}/K(z^p)}^{r}$, where $\epsilon$ is the $p$-degree  of  $ \omega$, because in general we have $\epsilon < s(d)=\min\{d,n(p-1)\}$ and the proof of the corollary will fail.  
\end{remark}

\begin{prop}\label{P:porco}

Let $\omega\in \Omega^r_{K[z]/K}$ be a polynomial $r$-form. If there are invariant irreducible polynomials $F_1$,..., $F_m$ in $K[z]$ and constants $\delta_1$,..., $\delta_m$ in $\mathbb Z$ (or in the prime sub-field $\mathbb Z_p$ of $K$ if char$(K)=p>0$) not all zero, such that
$$
\delta_1\Theta_{F_1}+...+\delta_m\Theta_{F_m}=0
$$
then $\omega$ has a rational (resp. polynomial) first integral.

\end{prop}
\begin{proof}
If we take the $1$-form $\eta=\delta_1\frac{dF_1}{F_1}+...+\delta_m\frac{dF_m}{F_m}$, then
$$
\omega\wedge \eta=\sum_{i=1}^m \delta_i\omega\wedge\frac{dF_i}{F_i}=\sum_{i=1}^m\delta_i\Theta_{F_i}=0
$$
Since $dF_i^{\delta_i}=\delta_iF_i^{\delta_i}\frac{dF_i}{F_i}$, then the rational function (or polynomial) $f=F_1^{\delta_1}...F_m^{\delta_m}$ is such that  $df=f\eta$ and then
$$
\omega\wedge df=f\omega\wedge \eta=0
$$
that is, $f$ is a first integral for $\omega$.
\end{proof}


\section{Darboux-Jouanolou Criterion}

\begin{lemma}\label{L:Dependence}
Let $\omega\in\Omega^r_{K(z)/K}$ be a rational $r$-form and let $\eta_1$,..., $\eta_r \in \mathcal E^*(\omega)$ be $K(z^p)$-linearly independent  closed  $1$-forms. If $\eta_1\wedge ... \wedge \eta_r=0$, then $\omega$ has a rational first integral.
\end{lemma}
\begin{proof}
Since $\eta_1\wedge ... \wedge \eta_r=0$, then  $\eta_1$,..., $\eta_r$ are linearly dependent over $K(z)$. Let $m<r$ be the largest positive integer such that $\eta_1$,..., $\eta_m$ are linearly independent over $K(z)$. Then
$$
\eta_{m+1}=\sum_{i=1}^m f_i\cdot \eta_i
$$
with $f_1$,..., $f_m\in K(z)$.

The $1$-forms $\eta_1$,..., $\eta_m$, $\eta_{m+1}$ are $K(z^p)$-linearly independent, hence there exists $j_0\in \{1,...,m\}$ such that $df_{j_0}\neq 0$, that is, $f_{j_0}\in K(z)-K(z^p)$

By hypothesis, $\eta_i$ is closed for $i=1,...,m+1$, then
$$
0=\sum_{i=1}^m df_i\wedge \eta_i
$$
Multiplying the above expression by $\eta_1\wedge...\wedge \widehat{\eta_{j_0}}\wedge ...\wedge \eta_m$, we obtain
$$
0=\sum_{i=1}^m df_i\wedge \eta_i\wedge\eta_1\wedge ... \wedge\widehat{\eta_{j_0}}\wedge...\wedge \eta_m=(-1)^{{j_0}+1}df_{j_0}\wedge \eta_1\wedge ...\wedge \eta_m
$$

Since $\eta_1$,..., $\eta_m$ are $K(z)$-linearly independent, then there exist $g_1$,..., $g_m\in K(z)$ such that $df_{j_0}=\sum_{i=1}^m g_i\cdot \eta_i$ and, by hypothesis, for each $j=1$,..., $m$,   $ \omega \wedge \eta_j=0$. Therefore
$$
\omega\wedge df_{j_0}=\sum_{i=1}^m g_i \cdot \omega\wedge \eta_i=0,
$$
that is, $f_{j_0}$ is a rational first integral for $\omega$.

\end{proof}

\begin{lemma}\label{L:Independence}
Let $\omega\in\Omega^r_{K(z)/K}$ be a rational $r$-form, with $ 1 \leq r <n$, and let $\eta_1$,..., $\eta_r$, $\eta_{r+1}$, $\in \mathcal E^*(\omega)$ be $K(z^p)$-linearly independent closed $1$-forms. If $\eta_1\wedge ... \wedge \eta_r\neq 0$ and $\eta_2\wedge ... \wedge \eta_{r+1}\neq 0$, then $\omega$ has a rational first integral.
\end{lemma}
\begin{proof}
By hypothesis $\{\eta_1,...,\eta_r\}$ and $\{\eta_2,...,\eta_{r+1}\}$ are two sets of $K(z)$-linearly independent closed $1$-forms in $\mathcal E^*(\omega)$. Write $\tau_1=\eta_1\wedge ... \wedge \eta_r$ and $\tau_2=\eta_2\wedge ... \wedge \eta_{r+1}$.

By Lemma \ref{L:Base} there exist rational functions $h_1$ and $h_2$ such that $\tau_i=h_i\cdot \omega$, $i=1,2$. Writing $f=\frac{h_1}{h_2}$, we have $\tau_1={f} \cdot \tau_2
$. Hence
$$
0=\tau_1-\tau_1=\tau_1-{f} \cdot \tau_2=(\eta_1+(-1)^rf \cdot\eta_{r+1})\wedge\eta_2\wedge...\wedge \eta_r
$$
Moreover, using that $d\tau_1=d\tau_2=0$, we have
$$
0=d\tau_1=df\wedge \tau_2=h_2df\wedge \omega
$$
and then $\omega\wedge df=0$.

Now we have two possibilities for this rational function: $f \in K(z)-K(z^p)$ or $f \in K(z^p)$. If $f \in K(z)-K(z^p)$, then $f$ is a rational first integral for $\omega$. On the other hand, suppose that  $f \in K(z^p)$. By hypothesis, the closed $1$-forms $\eta_1$, $\eta_2$,..., $\eta_r \in \mathcal E^*(\omega)$ are $K(z^p)$-linearly independent, which imply that  $\mu:=\eta_1+(-1)^rf\cdot \eta_{r+1}$, $\eta_2$,..., $\eta_r \in \mathcal E^*(\omega)$ are also $K(z^p)$-linearly independent closed $1$-forms. Then we conclude the existence of rational first integral for $\omega$ by Lemma \ref{L:Dependence}.


\end{proof}

\begin{thm}[Theorem A]\label{T:Main}
Let $\omega\in \Omega^r_{K[z]/K}$ be a polynomial $r$-form of degree $d$ and $N=\dim_{K(z^p)}(\Omega_{\mathcal L_{d-1}/K(z^p)}^{r+1})$. If $\omega$ admits $N+r+1$ invariant polynomials, then $\omega$ has a rational first integral.

\end{thm}
\begin{proof}
Suppose that the polynomial $r$-form $\omega$ admits $N+r+1$ invariant irreducible polynomials $F_1$,..., $F_{N+r+1}$. Then, for each $i\in \{1,2,...,N+r+1\}$, there exists $\Theta_{F_i}\in\Omega_{d-1}^{r+1}(K^n)$ such that
$$
\omega\wedge dF_i=F_i\cdot \Theta_{F_i}
$$

Since $\dim_{K(z^p)}(\Omega_{\mathcal L_{d-1}/K(z^p)}^{r+1})=N$, for each $s=1$,..., $r+1$, we can choose $\lambda_s^s$,..., $\lambda_{s+N}^s\in K(z^p)$ such that $\lambda_s^s\neq 0$ and
$$
\sum_{j=s}^{s+N} \lambda_j^s\cdot \Theta_{F_j}=0
$$

For each $s=1$,..., $r+1$, we can define the logarithmic  $1$-form
$$
\eta_s=\sum_{j=s}^{s+N} \lambda_j^s\cdot \frac{dF_j}{F_j}
$$
and then $\omega\wedge \eta_s=0$.

Define
$$
\tau_1=\eta_1\wedge ... \wedge \eta_r
$$
and
$$
\tau_2=\eta_2\wedge ... \wedge \eta_{r+1}
$$

By Lemma \ref{L:Logaritmicas}, $\{ \eta_1, ..,\eta_r \}$ and $\{ \eta_2, ..,\eta_{r+1} \}$ are two sets of logarithmic $1$-forms linearly independent over $K(z^p)$. Then:

\begin{itemize}
\item[(i)] if $\tau_1=0$ or $\tau_2=0$, we obtain a rational first integral for $\omega$ by Lemma \ref{L:Dependence};

\item[(ii)] if $\tau_1\neq 0$ and  $\tau_2\neq 0$, we obtain a rational first integral for $\omega$ by Lemma \ref{L:Independence}.

\end{itemize}

\end{proof}

We say that a polynomial $F$ is {\it invariant} by the vector field $X$ if $dF\neq 0$ and $X(F)$ is divisible by $F$. If $f$ is a rational function, we say that $f$ is a {\it rational first integral} for $X$ if $df\neq 0$ and $X(f)=0$.

We can associate to the vector field $X=\sum_{i=1}^n P_i \frac{\partial}{\partial z_i}$ the $(n-1)$-form
$$
\omega_X=\sum_{i=1}^n (-1)^{i+1} P_i dz_1\wedge ...\wedge \widehat{dz_i}\wedge ... \wedge dz_n
= i_X(dz_1 \wedge...\wedge dz_n),
$$
that is, $\omega_X$ is the {\it rotational} of $X$. Since
$$
\omega_X\wedge df=X(f)dz_1\wedge ... \wedge dz_n
$$
we have that $f$ (or $F$) is a rational first integral (or an invariant polynomial) for $X$ if and only if it is the same for $\omega_X$.

\begin{cor}\label{Cor: Field}
Let $X$ be a polynomial vector field of degree $d$ on $K^n$ and $N=\dim_{K(z^p)}(\Omega_{\mathcal L_{d-1}/K(z^p)}^{n})$. If $X$ admits $N+n$ invariant irreducible polynomials, then $X$ has a rational first integral.
\end{cor}
\begin{proof}
Let $X=\sum_{i=1}^n P_i \frac{\partial}{\partial z_i}$ be a polynomial vector field of degree $d$. Then $F$ is $X$-invariant if and only if $F$ is $\omega_X$-invariant. By Theorem \ref{T:Main}, $ \omega_X$ has a rational first integral $f\in K(z)$, which will be also a rational first integral for $X$.  

\end{proof}

Now we will prove that if $\omega$ has a rational first integral, then $\omega$ has infinitely many invariant polynomials. We will denote the {\it greatest common factor} between two polynomials $P$, $Q\in K[z]$ by $gcd(P,Q)$. In this way, $P$ and $Q$ are relatively primes if and only if $gcd(P,Q)=1$.

\begin{lemma}
Let $f=\frac{P}{Q}\in K(z)$ be a rational function such that $df\neq 0$. Let $H\in K[z]$ be a polynomial. If  $\lambda =(HP)^p\in K[z^p]$ then $P+\lambda Q$ has a non-$\partial$-constant factor $G$ with $gcd(H,G)=1$.
\end{lemma}
\begin{proof}
If $\lambda=(HP)^p$, then
$$
P+\lambda Q=P+H^pP^pQ=P(1+H^pP^{p-1}Q)
$$
Hence, if $G=1+H^pP^{p-1}Q$, then $gcd(H, G)=1$ and
$$
dG=d(1+H^pP^{p-1}Q)=H^pd(P^{p-1}Q)=H^p(P^{p-1}dQ+(p-1)P^{p-2}QdP)=
$$
$$
=H^p(P^{p-1}dQ-P^{p-2}QdP)=H^pP^{p-2}(PdQ-QdP)=-H^pP^{p-2}Q^2df\neq 0
$$

\end{proof}

\begin{lemma}\label{L:InfiniteFactors}
Let $f=\frac{P}{Q}\in K(z)$ be a rational function such that $df\neq 0$. Then there are infinitely many non-$\partial$-constant irreducible factors of polynomials $P+\lambda Q \in K[z]$, $\lambda\in K[z^p]$.
\end{lemma}
\begin{proof}
Let $H_1$,..., $H_m$ be non-$\partial$-constant irreducible factors of $P+\lambda_1 Q $,..., $P+\lambda_m Q$, respectively. Define $H=H_1\cdot ...\cdot H_m$. By the above lemma, there is $\lambda\in K[z^p]$ such that $P+\lambda Q$ has a non-$\partial$-constant factor $G$ such that $gcd(H,G)=1$. Therefore there is a non-$\partial$-constant irreducible factor $H_{m+1}$ of $G$. This proves the lemma.

\end{proof}

\begin{remark}
For any field $K$, if $P$, $Q \in K[z]$ are two coprime polynomials,
$$
\lambda_1\neq \lambda_2\in K\Rightarrow gcd(P+\lambda_1Q,P+\lambda_2Q)=1
$$
This imply the existence of infinitely many factors as claimed by Lemma \ref{L:InfiniteFactors} if $K$ is an infinite field (in particular if $K$ has characteristic $0$). But this property does not work for $\lambda_1\neq \lambda_2\in K[z^p]$. For example, suppose that $K$ has characteristic $p=2$. If $P=z_1+z_2^3=z_1+z_2^2\cdot z_2$ and $Q=z_2$, then $P$ and $Q$ are irreducible polynomials with $gcd(P,Q)=1$. If $\lambda_1=z_2^2$, $\lambda_2=(z_1+z_2)^2 \in K(z^2)$,
$$
F_1=P+\lambda_1 Q=z_1
$$
and
$$
F_2=P+\lambda_2 Q=z_1(1+z_1z_2)
$$
then $gcd(F_1, F_2)=z_1$. 
\end{remark}

\begin{thm}\label{T:InfiniteInvariants}
Let $\omega\in \Omega^r_{K[z]/K}$ be a polynomial $r$-form over a field $K$. Then $\omega$ has a rational first integral if and only if it has an infinite number of invariant irreducible polynomials. 

\end{thm}

\begin{proof}
Suppose that $f=\frac{P}{Q}$ is a rational first integral for $\omega$ and $gcd(P,Q)=1$. We have $\omega \wedge d(\frac{P}{Q})=0$ and in particular
$$
\omega \wedge (QdP-PdQ)=0
$$
which imply $Q\omega\wedge \frac{dP}{P}=\omega\wedge dQ$ and $P\omega\wedge \frac{dQ}{Q}=\omega\wedge dP$. Since $P$ and $Q$ are coprimes, then $\omega\wedge \frac{dP}{P}$ and $\omega\wedge \frac{dQ}{Q}$ are polynomial $(r+1)$-forms.

Consider $F_{\lambda}:=P+\lambda Q$, where $\lambda \in K[z^p]$. Then
$$
\omega\wedge dF_{\lambda}=\omega \wedge (dP+\lambda dQ)=\omega \wedge dP+\lambda \omega \wedge dQ=\omega \wedge dP+\lambda Q\omega\wedge \frac{dP}{P}=
$$
$$
=\omega\wedge (P\frac{dP}{P}+\lambda Q\frac{dP}{P})=(P+\lambda Q)\cdot \omega\wedge \frac{dP}{P}=F_{\lambda}\cdot\omega \wedge \frac{dP}{P}=F_{\lambda}\cdot\omega \wedge \frac{dQ}{Q}
$$
Therefore $\omega\wedge \frac{dF_{\lambda}}{F_{\lambda}}$ is a polynomial $(r+1)$-form, that is, for every $\lambda \in K[z^p]$ the non-$\partial$-constants irreducible factors of $F_{\lambda}$ are invariant. By the above lemma, we conclude the existence of infinitely many invariant polynomials for $\omega$.

Reciprocally, if there are infinitely many invariant polynomials we can apply Theorem \ref{T:Main} to conclude the existence of rational first integral for $\omega$.

\end{proof}


\section{Integrating Factors and Darboux Method}

\begin{definition}
Let $\omega\in \Omega^r_{K[z]/K}$ be a polynomial $r$-form. A rational function $f$ is said to be an {\it integration factor} of $\omega$ if $f\omega$ is closed, that is, $d(f\omega)=0$.
\end{definition}

The above definition of integration factor is different from those we find in the books of elementary differential equations in the plane (for example \cite{BDP}, Chapter 2, p. 99), where it is required $f\omega$ to be exact, that is, $f\omega=d\eta$. In characteristic $0$ the existence of a rational such $\eta$ is equivalent to the condition $d(f\omega)=0$ (Poincaré's Lemma), but that does not work in positive characteristic. The relationship  between closed and exact rational differential forms over fields of characteristic $p>0$ is explored in our pre-print \cite{Santos2}. 

\begin{prop}\label{P:Factor1}
Let $\omega\in \Omega^r_{K[z]/K}$ be a polynomial $r$-form. If there are invariant irreducible polynomials $F_1$,..., $F_m$ in $K[z]$ and $\partial$-constants $\lambda_1$,..., $\lambda_m$ in $K(z^p)^*$ such that
$$
\lambda_1 \Theta_{F_1}+\lambda_2 \Theta_{F_2}+...+\lambda_m \Theta_{F_m}=d\omega
$$
then there is a logarithmic $1$-form $\eta\neq 0$ such that $\omega\wedge \eta= d\omega $.
\end{prop}
\begin{proof}
Consider the logarithmic $1$-form
$$
\eta=\lambda_1 \frac{dF_1}{F_1}+\lambda_2 \frac{dF_2}{F_2}+...+\lambda_m \frac{dF_m}{F_m}\neq 0
$$
Then
$$
\omega\wedge \eta=\sum_{i=1}^m\lambda_i\omega \wedge \frac{dF_i}{F_i}=\sum_{i=1}^m\lambda_i \Theta_{F_i}=  d\omega
$$
\end{proof}

\begin{prop}\label{P:Factor2}
Let $\omega\in \Omega^r_{K[z]/K}$ be a polynomial $r$-form. If there are invariant irreducible polynomials $F_1$,..., $F_m$ in $K[z]$ and constants $\delta_1$,..., $\delta_m$ in $\mathbb Q$ (or in the prime sub-field $\mathbb Z_p$ of $K$ if $char(K)=p>0$), not all zero, such that
$$
\delta_1 \Theta_{F_1}+\delta_2 \Theta_{F_2}+...+\delta_m \Theta_{F_m}= d\omega
$$
then $\omega$ has a rational (respectively polynomial)   integration factor.
\end{prop}
\begin{proof}
As in the previous proposition, define the logarithmic $1$-form
$$
\eta=\delta_1 \frac{dF_1}{F_1}+\delta_2 \frac{dF_2}{F_2}+...+\delta_m \frac{dF_m}{F_m}\neq 0
$$
Up to multiplying each $\delta_1$,..., $\delta_m$ by $-1$ if necessary, we can suppose that  
$$
\omega\wedge \eta=\sum_{i=1}^m\delta_i\omega \wedge \frac{dF_i}{F_i}=\sum_{i=1}^m\delta_i \Theta_{F_i}= \sigma d\omega
$$
where $ \sigma= (-1)^{r+1} $, and hence $ (-1)^{r} \omega \wedge \eta = - d \omega $.
 
Define the rational (respectively polynomial) function
$$
f=F_1^{\delta_1}F_2^{\delta_2}...F_m^{\delta_m}
$$
By differentiation we have  $df=f\eta $, and  then

$$
d(f\omega)=df \wedge \omega+ f d\omega=(-1)^{r} \omega \wedge df + f d \omega=(-1)^{r} f\omega \wedge \eta + f d \omega= - f d\omega+ f d \omega=0 
$$

\end{proof}

\begin{thm}\label{T:Main4}
Let $\omega\in \Omega_d^r(K^n)$ be a rational $r$-form of degree $d$. If $\omega$ has $N=N_K(n,d-1,r+1)=\dim_{K(z^p)}(\Omega_{\mathcal L_{d-1}/K(z^p)}^{r+1})$ invariant irreducible polynomials $F_1$, $F_2$,..., $F_N$, then there is a logarithmic $1$-form $\eta\neq 0$ such that $\omega\wedge \eta=0$ or $\omega\wedge \eta=d\omega$. If $K=\mathbb Q$ or $K=\mathbb Z_p$, $p>0$, then $\omega$ has a rational first integral or a rational integration factor.
\end{thm}
\begin{proof} The cofactors associated to the invariant polynomials are differential $(r+1)$-forms of degree lass than or equal to $d-1$. By the natural inclusion ${\Omega}_{d-1}^{r+1}(K^n)\rightarrow\Omega_{\mathcal L_{d-1}/K(z^p)}^{r+1}$ and since we have $N=\dim_{K(z^p)}(\Omega_{\mathcal L_{d-1}/K(z^p)}^{r+1})$ cofactors associated to the polynomials $F_i$, $i=1$,..., $N$, then  that cofactors are linearly dependent or form a base to the $K(z^p)$-vector space $\Omega_{\mathcal L_{d-1}/K(z^p)}^{r+1}$. Hence there are two possibilities:

{\it Case $1$. The cofactors associated to the polynomials $F_i$, $i=1$,..., $N$, are $K(z^p)$-linearly dependent}.
Then there are $\partial$-constants $\lambda_1$,..., $\lambda_N$ in $K$, not all zero, such that
$$
\lambda_1 \Theta_{F_1}+\lambda_2 \Theta_{F_2}+...+\lambda_N \Theta_{F_N}=0
$$
hence for the logarithmic $1$-form
$$
\eta=\sum_{i=1}^N \lambda_i \frac{dF_i}{F_i}
$$
as in  Proposition  \ref{P:porco},  we have $\omega\wedge \eta=0$, and if $K=\mathbb Q$ or $K=\mathbb Z$, then $f=F_1^{\lambda_1}...F_N^{\lambda_N} $ defines a first integral for $\omega$.

{\it Case $2$. The cofactors associated to the polynomials $F_i$, $i=1$,..., $N$, form a base to the $K(z^p)$-vector space $\Omega_{\mathcal L_{d-1}/K(z^p)}^{r+1}$}. Then, as in  Proposition \ref{P:Factor2}, the $(r+1)$-form $d\omega$ can be written as a $K(z^p)$-linear combination of the cofactors $\Theta_{F_i}$ and we obtain a logarithmic $1$-form $\eta$ such that $\omega\wedge \eta=d\omega$, and if $K=\mathbb Q$ or $K=\mathbb Z$, then $f=F_1^{\lambda_1}...F_N^{\lambda_N}$ defines an integrating factor for $\omega$.

\end{proof}

\begin{example}[Lotka-Volterra Equation]
Consider the $1$-form
$$
\omega=y(\gamma - \delta x)dx+x(\alpha-\beta y)dy
$$
We can easily see that the polynomials $F=x$ and $G=y$ are invariant by $\omega$, and their cofactors are
$$
\Theta_x=\omega\wedge \frac{dx}{x}=(\beta y-\alpha )dx\wedge dy
$$
and
$$
\Theta_y=\omega\wedge \frac{dy}{y}=(\gamma y-\delta )dx\wedge dy
$$
Taking the differential of $\omega$, we have
$$
d\omega=[(\alpha - \beta y)-(\gamma - \delta x)]dx\wedge dy=-\Theta_x-\Theta_y
$$
and therefore
$$
d\omega = \omega \wedge (-\frac{dx}{x}-\frac{dy}{y})
$$
and then $(xy)^{-1}$ is a rational  integration factor for $\omega$.
\end{example}

 

\begin{thm}[Theorem B]\label{T:Main2}
Let $\omega\in \Omega^r_{K[z]/K}$ be a polynomial $r$-form of degree $d$ and $N=N_K(n,d-1,r+1)=\dim_{K(z^p)}(\Omega_{\mathcal L_{d-1}/K(z^p)}^{r+1})$. If $\omega$ admits $N+r$ invariant irreducible polynomials, then $\omega$ has a rational first integral or a rational integration factor.

\end{thm}

\begin{proof}
Suppose that the polynomial $r$-form $\omega$ admits $N+r$ invariant irreducible polynomials $F_1$,..., $F_{N+r}$. Then, for each $i\in \{1,2,...,N+r\}$, there exists $\Theta_{F_i}\in\Omega_{d-1}^{r+1}(K^n)$ such that
$$
\omega\wedge dF_i=F_i\cdot \Theta_{F_i}
$$

Since $\dim_{K(z^p)}(\Omega_{\mathcal L_{d-1}/K(z^p)}^{r+1})=N$, for each $s=1$,..., $r$, we can choose $\lambda_s^s$,..., $\lambda_{s+N}^s\in K[z^p]$, not all zero, such that $\lambda_s^s\neq 0$ and
$$
\sum_{j=s}^{s+N} \lambda_j^s\cdot \Theta_{F_j}=0
$$

For each $s=1$,..., $r$, we can define the logarithmic $1$-form
$$
\eta_s=\sum_{j=s}^{s+N} \lambda_j^s\cdot \frac{dF_j}{F_j}
$$
and then $\omega\wedge \eta_s=0$. Define
$$
\tau=\eta_1\wedge ... \wedge \eta_r
$$
By Lemma \ref{L:Logaritmicas}, $\eta_1, ...,\eta_r$ are linearly independent over $ K(z^p)$ and: 

\begin{itemize}
\item[(i)] if $\tau=0$ then, by Lemma \ref{L:Dependence}, we obtain a rational first integral for $\omega$;

\item[(ii)] if $\tau\neq 0$ then,  by Lemma \ref{L:Base}, there is a rational function $h$ such that $\omega=h\cdot\tau$; therefore, since $\tau$ is closed, with $f=\dfrac{1}{h}$, we have
$$
d(f\cdot\omega)=d\tau=0 ,
$$
that is,  $f$ is an integration factor for $\omega$.

\end{itemize}

\end{proof}

\begin{remark}
If $\omega\in \Omega^r_d(K^n)$ has a rational first integral, then $\omega$ admits infinitely many invariant polynomials, so $\omega$ also has a rational integration factor. 
\end{remark}


\bibliographystyle{amsplain}

\end{document}